\documentclass[11pt]{article}
\newcommand{\eps}{\varepsilon}
\usepackage{amsmath,amssymb,amsthm}
\usepackage{fullpage}
\usepackage{paralist}
\usepackage{comment}
\usepackage{enumitem}
\usepackage{algorithm}
\usepackage[noend]{algorithmic} 
\usepackage[bottom]{footmisc}

\def\Comment#1{\textsf{\textsl{$\langle\!\langle$#1\/$\rangle\!\rangle$}}}
\def\set#1{\{ #1 \}}
\def\polylog{\operatorname{polylog}}
\def\mathsc#1{\text{\textsc{#1}}}

\newcommand{\churnlimit}{n/\polylog{n}}
\let\oldmarginpar\marginpar
\renewcommand\marginpar[1]{\-\oldmarginpar[\raggedleft\footnotesize #1]%
{\raggedright\footnotesize #1}}

\newcommand{\I}{\mathcal{I}}
\newcommand{\G}{\mathcal{G}}
\newcommand{\bG}{\bar{\mathcal{G}}}
\newcommand{\bg}{\bar{G}}
\newcommand{\bv}{\bar{V}}
\newcommand{\be}{\bar{E}}

\newcommand{\m}{\mathsc{m}}

\renewcommand{\d}{\mathsc{d}}
\renewcommand{\t}{\tau}

\newcommand{\Comm}{\textsc{Com}}
\newcommand{\Com}{\textsc{Com}}

\newcommand{\Exp}[1]{\text{E}\left[#1\right]}
\newcommand{\Prob}[1]{\text{Pr}\left[#1\right]}

\newcommand{\core}{\mathsc{{Core}}}
\newcommand{\Core}{\core}

\renewcommand{\ge}{\geqslant}
\renewcommand{\le}{\leqslant}
\renewcommand{\geq}{\geqslant}

\newtheorem{definition}{Definition}

\newtheorem{theorem}{Theorem}
\newtheorem{lemma}{Lemma}

\newtheorem{corollary}{Corollary}

\newboolean{short}
\setboolean{short}{false}

\newcommand{\shortOnly}[1]{\ifthenelse{\boolean{short}}{#1}{}}
\newcommand{\onlyShort}[1]{\ifthenelse{\boolean{short}}{#1}{}}
\newcommand{\longOnly}[1]{\ifthenelse{\boolean{short}}{}{#1}}
\newcommand{\onlyLong}[1]{\ifthenelse{\boolean{short}}{}{#1}}

\begin{document}
\onlyShort{%
\conferenceinfo{SPAA'13,} {June 23--25, 2013, Montr\'eal, Qu\'ebec, Canada.}
\CopyrightYear{2013}
\crdata{978-1-4503-1572-2/13/07}
\clubpenalty=10000
\widowpenalty = 10000
}

\onlyLong{
\begin{titlepage}
}
\title{Storage and Search in Dynamic Peer-to-Peer Networks}
\onlyShort{
\numberofauthors{6}
\author{
\alignauthor John Augustine\\
  \affaddr{Department of Computer Science and Engineering}\\
  \affaddr{Indian Institute of Technology Madras,}
  \affaddr{Chennai, India.}\\
  \email{augustine@cse.iitm.ac.in}
\alignauthor Anisur Rahaman Molla  \\
  \affaddr{Division of Mathematical Sciences}\\
  \affaddr{Nanyang Technological University}
  \affaddr{Singapore 637371}\\
  \email{anisurpm@gmail.com}
\alignauthor Ehab Morsy\\
  \affaddr{Department of Mathematics}\\
  \affaddr{Suez Canal University}\\
  \affaddr{Ismailia 22541, Egypt} \\
  \email{ehabmorsy@gmail.com}.
\and  %
\alignauthor Gopal Pandurangan\titlenote{Research supported in part by the following  grants: Nanyang Technological University grant M58110000, Singapore Ministry of Education (MOE) Academic Research Fund (AcRF) Tier 2 grant MOE2010-T2-2-082, and a grant from the US-Israel Binational Science Foundation (BSF).}\\
  \affaddr{Division of Mathematical Sciences}\\
  \affaddr{Nanyang Technological University}\\
  \affaddr{Singapore 637371}\\ 
  \email{gopalpandurangan@gmail.com}.  
\alignauthor Peter Robinson\titlenote{Research supported in part by the following  grants: Nanyang Technological University grant M58110000, Singapore Ministry of Education (MOE) Academic Research Fund (AcRF) Tier 2 grant MOE2010-T2-2-082}\\
  \affaddr{Division of Mathematical Sciences}\\
  \affaddr{Nanyang Technological University}
  \affaddr{Singapore 637371}\\ 
  \email{peter.robinson@ntu.edu.sg}.  
\alignauthor Eli Upfal\\
  \affaddr{Department of Computer Science}\\
  \affaddr{Brown University} \\ 
  \affaddr{Providence, RI 02912, USA}\\
  \email{eli@cs.brown.edu}
}
}
\onlyLong{
\author{John Augustine\thanks{Department of Computer Science and Engineering, Indian Institute of Technology Madras, Chennai, India.  
\hbox{E-mail}:~{\tt augustine@cse.iitm.ac.in}.}  
\and Anisur Rahaman Molla\thanks{Division of Mathematical
Sciences, Nanyang Technological University, Singapore 637371.
\hbox{E-mail}:~{\tt anisurpm@gmail.com, peter.robinson@ntu.edu.sg}} 
\and Ehab Morsy\thanks{\ Division of Mathematical
Sciences, Nanyang Technological University, Singapore 637371 and Department of Mathematics, Suez Canal University, Ismailia 22541, Egypt. 
\hbox{E-mail}:~{\tt ehabmorsy@gmail.com}.}  
\and Gopal Pandurangan\thanks{\ Division of Mathematical
Sciences, Nanyang Technological University, Singapore 637371 and Department of Computer Science, Brown University, Box 1910, Providence, RI 02912, USA.  \hbox{E-mail}:~{\tt gopalpandurangan@gmail.com}.  Research supported in part by the following  grants: Nanyang Technological University grant M58110000, Singapore Ministry of Education (MOE) Academic Research Fund (AcRF) Tier 2 grant MOE2010-T2-2-082,
and a grant from the US-Israel Binational Science
Foundation (BSF).} 
\and Peter Robinson$^{\dagger}$ 
\and Eli Upfal\thanks{Department of Computer Science, Brown University, Box 1910,
Providence, RI 02912, USA. \hbox{E-mail}:~{\tt eli@cs.brown.edu}}}
}

\date{}

\maketitle

\begin{abstract}
We study robust and efficient  distributed algorithms for searching, storing, 
and maintaining data  in dynamic Peer-to-Peer (P2P) networks. P2P
networks  are  highly dynamic networks  that experience heavy
node {\em  churn} (i.e., nodes  join and leave the network continuously over time). Our
goal  is to 
guarantee, despite high node churn rate,  that a large number of nodes in the network
can  store, retrieve, and maintain a large number of data items. 
Our main contributions are fast randomized  distributed algorithms  that
guarantee the above  with high probability even under
high {\em adversarial} churn.   In particular, we
present the following main results:
\begin{enumerate}
\item  A  randomized distributed
  search algorithm that with high probability guarantees that  searches from as many as $n - o(n)$ 
  nodes  ($n$ is the stable network size) succeed in ${O}(\log n )$-rounds despite ${O}(n/\log^{1+\delta} n)$ churn, 
  for any  small constant $\delta > 0$,  {\em per round}. 
  We assume that the churn is controlled by an
  oblivious adversary (that has  complete knowledge and control of what nodes join
  and leave and  at what time  and has unlimited computational power, but is
  oblivious to the random choices made by the algorithm).   
  \item  A storage and maintenance algorithm 
   that guarantees, with high probability,   data items  can be  
  efficiently stored (with only $\Theta(\log{n})$ copies of each data item) and
  maintained in a dynamic P2P network with churn rate up to ${O}(n/\log^{1+\delta} n)$ per round. 
  Our search algorithm together with  our storage and maintenance algorithm guarantees that
  as many as $n - o(n)$ nodes can efficiently store, maintain, and search  even
  under ${O}(n/\log^{1+\delta} n)$ churn  {\em per round}. Our  algorithms require
only polylogarithmic  in $n$  bits to be processed and sent (per round) by each node.
  \end{enumerate}
    
To the best of our knowledge, our algorithms are the first-known, fully-distributed storage and search algorithms that provably work under highly dynamic settings (i.e., high churn rates per step). 
Furthermore, they are  localized (i.e., do not require any global topological knowledge) and scalable.  
A technical contribution of this paper, which may be of independent interest, is  showing how random walks 
can be provably used to derive scalable distributed algorithms in dynamic networks with adversarial node churn.
\end{abstract}
\onlyShort{
\category{F.2.2}{Theory of Computation}[Analysis of Algorithms and Problem Complexity---Nonnumerical Algorithms and Problems]
\terms{Theory, Algorithms}

}

\onlyLong{\noindent {\bf Keywords:}  Peer-to-Peer network,
  Dynamic network, Search, Storage, Distributed algorithm, Randomized algorithm, Random Walks, Expander graph.}

\onlyLong{
\end{titlepage}
}
\section{Introduction}
Peer-to-peer (P2P) computing is emerging as one of the key networking technologies in recent years with many application systems, e.g., Skype, BitTorrent, Cloudmark,  CrashPlan, Symform etc. 
For example,  systems such as CrashPlan \cite{crashplan} and Symform \cite{symform} are relatively recent P2P-based storage services that  allow data to be stored and retrieved
among peers \cite{p2p-storage-survey}. Such data sharing among peers  avoids
costly centralized storage and retrieval, besides being inherently scalable to millions of peers.  However, many of these systems are not fully P2P; they also use dedicated centralized servers in order to guarantee high availability of data --- this  is necessary due to the highly dynamic and unpredictable nature of P2P.  Indeed, a key reason for the lack of fully-distributed P2P systems is  the difficulty in designing highly robust algorithms for large-scale dynamic P2P networks. 

P2P networks  are  highly dynamic  networks  characterized by  high degree of node {\em  churn}  --- i.e., nodes continuously join and leave the network. Connections (edges) may be added or deleted at any time and thus the topology changes very dynamically. In fact,  
measurement studies of real-world P2P networks~\cite{FPJKA07,SGG02,SW02, SR06} show that the churn rate is quite high:
nearly 50\% of peers in real-world networks can be replaced within an hour. (However,  despite a large churn rate, these studies  also show  that the
total number of peers in the network is relatively {\em stable}.) 
P2P algorithms have been proposed
for a wide variety of  tasks such as data storage and retrieval \cite{p2p-storage-survey, dr01, past01, malugo10, hasan05}, collaborative
filtering~\cite{Canny02}, spam detection~\cite{Cloudmark},
data mining~\cite{DBGWK06}, worm detection and
suppression~\cite{MS05,VAS04}, privacy protection of archived
data~\cite{GKLL09}, and recently, for cloud computing services as well \cite{symform, ozlap12}.  
 However, all algorithms  proposed for these
problems have no theoretical guarantees of being able to
work in a dynamically changing network with a very high churn rate, which can be  as much as  linear (in the network size)  per round. This is  a major bottleneck in implementation and
wide-spread use of P2P systems.

In this paper, we take a step towards  designing provably robust  and scalable algorithms for large-scale dynamic P2P networks. 
In particular, we focus on the fundamental problem of  storing, maintaining, and searching  data in  P2P networks.  Search in P2P  networks
is a well-studied fundamental application with a large body of work in the last decade or so, both in theory and practice (e.g., see the survey \cite{LCP+04}). 
While many P2P systems/protocols have been proposed  for efficient search and storage of data  (cf. Section \ref{sec:related}), a major drawback of almost all these  is the lack of algorithms that  work with provably guarantees under a large amount of churn per round.  The problem is especially challenging since the goal 
is to guarantee that  almost all nodes\footnote{In sparse, bounded-degree networks, as assumed in this paper,
an adversary can always isolate some number of nodes due to churn, hence ``almost all"  is
the best one can hope for in such networks.}  are able to efficiently store, maintain, and retrieve data,  even under high churn rate.   
In such a highly dynamic setting, it is non-trivial to even just  store data in a persistent manner; the churn can simply remove a large fraction of nodes in just one time step.  On the other hand,
it is costly to replicate too many copies of a data item  to guarantee persistence. Thus the challenge
is to use as little storage as possible and maintain the data for a long time, while at the same time designing efficient search algorithms that find
the data  quickly, despite high churn rate. Another complication to this challenge is designing
algorithms that are also scalable, i.e., nodes that process and send only a small number of (small-sized) messages per round.

\subsection{Our Main Results}
 We provide  a
rigorous theoretical framework for the design and analysis of storage, maintenance, and retrieval algorithms for
highly dynamic distributed systems with churn. We briefly describe the key ingredients of our model here. (Our model is described in detail in  Section
\ref{sec:model}). Essentially, we model a P2P network as a bounded-degree expander graph   whose topology --- both nodes and edges ---  can change arbitrarily from round to round and is controlled by an adversary. However, we assume
that the total number of nodes in the network is stable.  The number of node
changes {\em per round} is called the {\em churn rate} or {\em churn limit}. We
consider a churn rate of up to some ${O}(n/\log^{1+\delta} n)$\footnote{Throughout this paper, we use $\log$ to represent natural logarithm unless explicitly specified otherwise.}, where $\delta > 0$ is any small constant and $n$ is the stable network size. Note that our model is quite general in the
sense that we only assume that the topology is an expander at every step; no other special
properties are assumed. Indeed,  expanders have been used extensively to
model dynamic P2P networks\footnote{A number of works on static networks  have used expander graph topologies to  solve the agreement and related problems \cite{KKKSS10, DPPU88, Upfal94}. Here we show that similar expansion properties are beneficial in the more challenging setting of dynamic networks (cf. Section \ref{sec:related}).}
    in which the expander property is preserved under insertions and
deletions of nodes (e.g., \cite{APRU12,LS03,PRU01}).   Since we do not make assumptions on how the topology is preserved,
 our model is applicable to all such expander-based networks. 
  (We note that  various prior work on  dynamic network models (e.g., \cite{AKL08, kuhn+lo:dynamic, APRU12,  SMP12}) make similar assumptions on preservation of  topological  properties  --- such as connectivity, high expansion etc. ---
 at every step  under  insertions/deletions --- cf. Section \ref{sec:related}.  The issue
 of how such properties are preserved are abstracted away from the model, which allows one to focus on
 the dynamism. Indeed, this abstraction has been a feature of most dynamic
 models e.g., see the survey  of \cite{santoro}.)

 Our main contributions are efficient randomized  distributed algorithms for searching, storing, and maintaining data  in dynamic P2P networks. Our algorithms  succeed with high probability (i.e., with probability $1 - 1/n^{\Omega(1)}$, where $n$ is the stable network size)) even under high
adversarial churn in a polylogarithmic number of
rounds.   In particular, we
present the following results (the precise theorem statements are given in Section \ref{sec:storage}):
\begin{compactenum}
\item (cf.\ Theorem \ref{thm:storage})  A storage and maintenance algorithm  that guarantees, with high probability, that  data items  can be  
  efficiently stored (with only $\Theta(\log{n})$ copies of each data item \footnote{Using erasure coding techniques,
  the number of bits stored can be reduced even further, so as to incur only a constant factor overhead. We discuss
  this in\onlyLong{ Section \ref{sec:erasure-codes}.}\onlyShort{ in the full version of the paper.}}) and
  maintained in a dynamic P2P network with churn rate up to ${O}(n/\log^{1+\delta} n)$ {\em per round},  assuming that the churn is controlled by an
  oblivious adversary (that has  complete knowledge and control of what nodes join
  and leave and  at what time  and has unlimited computational power, but is
  oblivious to the random choices made by the algorithm).
    \item 
  (cf.\ Theorem \ref{thm:retrieval})  A  randomized distributed
  search algorithm that with high probability guarantees that  searches from as many as $n - o(n)$ 
  nodes  succeed in ${O}(\log n)$-rounds under up to ${O}(n/\log^{1+\delta} n)$ churn  {\em per round}.
    Our search algorithm together with the storage and maintenance algorithm guarantees that
  as many as $n -o(n)$ nodes can efficiently store, maintain, and search  even
  under ${O}(n/\log^{1+\delta} n)$ churn  {\em per round}. 
  Our  algorithms require
only polylogarithmic in $n$ bits to be processed and sent (per round) by each node.
     \end{compactenum}
    
To the best of our knowledge, our algorithms are the first-known,
fully-distributed storage and search algorithms that work under
highly dynamic settings (i.e., high churn rates per step). Furthermore,
they are  localized (i.e., do not require any global topological
knowledge) and scalable.  
\subsection{Technical Contributions}
We derive techniques   (cf.\ Section
\ref{sec:soup}) for doing  scalable distributed computation in highly dynamic  networks. 
In such networks,
 we would like  distributed algorithms to  work correctly and efficiently, and terminate even in networks that keep changing continuously over time (not assuming any eventual stabilization). 
The main technical tool that we use  is random walks. Flooding techniques (which proved useful 
in solving the agreement problem under high adversarial churn \cite{APRU12}) are not useful for search as they generate lot
of messages and hence are not scalable.
  We note that random walks have been used before to perform search in P2P networks (e.g., \cite{aravind,shenker,zs+06,gms05}) as well
  for other applications  such as sampling (e.g., \cite{ganesh07,brahms}), but these are not applicable 
  to  dynamic  networks with large  churn rates. 

One of the main technical contributions of this paper is showing how  random walks can be used in a dynamic network with high adversarial {\em node churn} (cf. Section \ref{sec:soup}).
The basic idea is quite simple and is as follows. All nodes  generate tokens   (which contain the source node's ids) and send it via random walks continuously over time. These random walks, 
once they ``mix" (i.e., reach close to the stationary distribution), reach essentially ``random"
destinations in the network; we (figuratively)  call these simultaneous random walks
as a {\em soup of random walks}. Thus the destination nodes  receive a steady stream of tokens from essentially random nodes, thereby allowing them to sample nodes uniformly from the network. While this  is easy to establish in a static network, it is no longer true
in a dynamic network with adversarial churn --- the churn can cause many
random walks to be lost  and also might introduce  bias.
We show a technical result called the ``Soup Theorem" (cf. Theorem \ref{thm:soup}) that shows that ``most" random walks do mix (despite large adversarial churn) 
and have the usual desirable properties as in a static network. We use the Soup theorem 
crucially in   our search, storage, and maintenance algorithms.  We note that
our technique can handle churn only up to $\churnlimit$. Informally, this is due to the fact that
at least  $\Omega(\log n)$ rounds are needed for the random walks to mix, before which any non-trivial computation can be performed. This seems to be a fundamental limitation of our random walk based method. We come close to this limit in that we allow churn to be as high as $O(n/\log^{1+\delta} n)$ for {\em any} fixed $\delta > 0$.

Another technique that we use as a building block in our algorithms is construction and maintenance
of (small-sized) committees. A  committee is a  clique of small ($\Theta(\log n)$) size composed  of essentially
``random" nodes.  We show how such a committee can be efficiently constructed, and more importantly,
{\em maintained} under large churn.  A committee can be used to ``delegate" a storage or a search operation;
its small size guarantees scalability, while its persistence guarantees that the operation will complete successfully
despite churn.
Our techniques (the Soup Theorem and committees) can be useful in other distributed applications as well, e.g., leader election.

\onlyShort{Due to space limitation, full proofs and omitted details are in the full version of the paper.}

\subsection{Related Work}\label{sec:related}

There has been significant prior work in designing P2P networks that
are provably robust to a large number of Byzantine
faults (e.g., see ~\cite{FS02,HK03,NW03,Scheideler05,AS09}). These  focus on robustly enabling storage and retrieval of data items under adversarial nodes. However, these algorithms will not work in a highly dynamic setting with {\em large,  continuous, adversarial} churn (controlled by an all-powerful adversary that  has full control of the network topology, including full knowledge and control of what nodes join
  and leave and  at what time  and has unlimited computational power). Most prior works develop algorithms that will work under the assumption that 
the network will eventually stabilize and stop changing. (An important aspect of our algorithms  is that they will work and terminate correctly even when the
  network keeps continually changing.)  There has been a lot of work on P2P  algorithms for maintaining desirable properties (such as  connectivity, low diameter, bounded degree) under churn (see e.g., \cite{PRU01, JP12, KSW10}, but these don't work under  large adversarial  churn rates.
  In particular, there has been very little work till date that rigorously addresses distributed computation in dynamic P2P networks under high node churn. The work (\cite{KSSV06})
raises the open question of whether one can design robust P2P  protocols that
can work in highly dynamic networks with a large adversarial churn.  The recent work of \cite{APRU12} was one of the first to address the above question; its focus was on solving the fundamental agreement problem in P2P networks with very large adversarial churn. However,
the paper does not address the problem of search and storage, which was a problem left open
in \cite{APRU12}.

\onlyLong{
There has been significant work in the design  of P2P systems for doing efficient search. These can be classified into two categories --- (1) Distributed Hash Table(DHT)-based schemes (also called ``structured" schemes) and (2) unstructured schemes; we refer to \cite{LCP+04} for a detailed survey.
However much of these systems have no provable performance guarantees under large adversarial churn.
DHT schemes  create a fully decentralized index that maps data
items to peers  and allows a peer to search for a data item
efficiently without flooding.    In unstructured networks, there is no relation
between the data identifier and the peer where it resides. 
 There also have been a lot of work on search in unstructured network topologies, see e.g., \cite{shenker, aravind} and the references therein.  Our algorithms assume an unstructured network.
}

There has been works  on building fault-tolerant Distributed Hash Tables (which
are classified as ``structured" P2P networks unlike ours which are ``unstructured" e.g., see \cite{PRU01}) under  different deletion models --- adversarial deletions and stochastic deletions. The structured P2P network described by Saia {\em et
al.}~\cite{SFG+02} 
guarantees that a large number of data items are available even if a large fraction of
{\em arbitrary} peers are deleted, under the assumption that, at any time, the number
of peers deleted by an adversary must be smaller than the number of peers joining.
 \onlyLong{Kuhn et al. consider in  \cite{KSW10} that up to $O(\log n)$ nodes (adversarially chosen) can crash or join per constant number of time
steps. Under this amount of churn, it is shown in \cite{KSW10} how to maintain a low peer
degree and bounded network diameter in P2P systems by using the hypercube and
pancake topologies.  Scheideler and Schmid show in \cite{SS09} how to maintain a
distributed heap that allows join and leave operations and, in addition, is
resistent to Sybil attacks. A robust distributed implementation of a distributed
hash table (DHT) in a P2P network is given by \cite{AS09}, which can withstand
two important kind of attacks: adaptive join-leave attacks and adaptive
insert/lookup attacks by up to $\eps n$ adversarial peers.  This paper
assumes that the good nodes always stay in the system and the adversarial nodes
are churned out and in, but the  {\em algorithm} determines where to insert the new nodes.
Naor and Weider \cite{NW03} describe a simple DHT scheme that is robust under the following simple random deletion model --- each node can fail independently with probability $p$. They show that their scheme can guarantee logarithmic degree, search time, and message complexity if $p$ is sufficiently small.
\onlyLong{Hildrum and Kubiatowicz \cite{HK03} describe how to modify two popular DHTs, Pastry \cite{Pastry} and Tapestry \cite{ZKJ01} to tolerate random deletions.
 Several DHT schemes (e.g., \cite{SM+01,RF+01,Koorde}) have been  shown to be robust under the simple random deletion model mentioned above.}
 There also have been works on designing fault-tolerant storage systems in a dynamic setting
 using quorums (e.g., see \cite{dynamic-quorum03, NU-quorum05}). However, these do not apply to our model of continuous churn.}

\onlyLong{To address the unpredictable and often unknown
nature of network dynamics,~\cite{kuhn+lo:dynamic} study a model
in which the communication graph can change completely from one round
to another, with the only constraint being that the network is
connected at each round.  The model of~\cite{kuhn+lo:dynamic} allows
for a much stronger adversary than the ones considered in past work on
general dynamic
networks~\cite{awerbuch+bbs:route,awerbuch+l:flow,awerbuch+bs:anycast}.
The surveys of~\cite{kuhn-survey, dynamic-survey} summarizes recent work on dynamic
networks.}
 
The dynamic network model of~\cite{kuhn+lo:dynamic, AKL08, SMP12}
allows only edge changes from round to round while the nodes remain
fixed. %
In our work, we study a dynamic
network model  where {\em both} nodes and edges
can change by a large amount.  
Therefore, the framework we study in Section~\ref{sec:model} (and first introduced in \cite{APRU12}) is more general
than the model of \cite{kuhn+lo:dynamic}, as it is additionally applicable to dynamic
settings with node churn. We note that the works of \cite{AKL08,SMP12} study random
walks under a dynamic model where the nodes are fixed (and only edges change) and hence not applicable to systems with churn.  The surveys of~\cite{kuhn-survey, dynamic-survey} summarizes recent work on dynamic
networks. %

 Expander graphs and spectral properties have already been applied
extensively to improve the network design and fault-tolerance in
distributed computing in general (\cite{Upfal94,DPPU88,BBCES2006}) and P2P networks in particular \cite{KSSV06, APRU12}. The problem of achieving almost-everywhere agreement among nodes in P2P networks --- modeled as an  expander graph --- is
considered by King et al.\ in \cite{KSSV06} in the context of the leader
election problem. However,
the algorithm of \cite{KSSV06} does not work for dynamic networks. The work of \cite{APRU12}  addresses the agreement problem in a dynamic P2P network under an adversarial  churn model where the churn rates  can be very large, up to linear in the number of nodes in the network. It also crucially makes use of expander graphs.
\onlyShort{(More related work can be found in the full paper.)}

\section{Model and Problem Statement}\label{sec:model}

\onlyLong{\subsection{Dynamic Network Model}}
We consider a synchronous  dynamic network with churn represented by a dynamically changing graph whose edges represent connectivity in the network. Our model is similar to the one introduced in \cite{APRU12}.
The computation is structured into synchronous rounds, i.e., we assume that
nodes run at the same processing speed and any message that is sent by some node
$u$ to its (current) neighbors in some round $r\ge 1$ will be received by the end of $r$.
To ensure scalability, we restrict the number of bits sent per round by each
node to be polylogarithmic in $n$, the stable network size.
In each round, up to ${O}(n/\log^{1+\delta} n)$ nodes can be replaced by new nodes, for any small constant $\delta>0$.    
Furthermore, we allow the edges to change arbitrarily in each round, but the underlying graph must be a ${\d}$-regular non-bipartite expander graph (${\d}$ can be a constant). (The regularity assumption can be relaxed, e.g.,
it is enough for nodes to have approximately equal degrees, and our results can be extended.)
The churn and edge changes are made by an adversary that is
\emph{oblivious} to the state of the nodes.  (In particular, it does not know the random choices made by the nodes.)
More precisely, the dynamic network is represented by a sequence of graphs $\mathcal{G} = (G^0, G^1, \ldots)$. 
We assume that the adversary commits to this sequence of graphs before round 0, but the algorithm is unaware of the sequence.
Each $G^r = (V^r, E^r)$ has $n$ nodes. 
We require that for all $r \ge 0$, $|V^r \setminus V^{r+1}| = |V^{r+1} \setminus V^{r}| \le {O}(n/\log^{1+\delta} n)$.
Furthermore, each $G^r$ must be a $\d$-regular non-bipartite expander with a fixed  upper bound of $\lambda$  on the second largest eigenvalue in absolute value. 

A node $u$ can communicate with any node $v$ if $u$ knows the id of $v$.\footnote{This is a typical assumption in the context of P2P  networks, where a node can establish communication with another node if it knows the other node's IP address.} When a new node joins 
the network, it has only knowledge of the ids  of  its current neighbors in the network and thus 
can communicate with them. 
We note that communication can be highly unreliable  due to churn, since  when $u$ sends a message to  $v$ there is no guarantee that $v$ is still in the network. 
However, each node in the network is guaranteed to have $\d$ neighbors in the network at any round with whom it can reliably communicate in that round. We note that random walks always use  the neighbor edges.

The network is synchronous, so nodes operate under a common clock. 
The following sequence  events occur in each round or  time step $r$. Firstly, the adversary makes the necessary changes to the network, so the algorithm is presented with graph $G^r$. So each node becomes aware of its neighbors in $G^r$. Each node then exchanges messages with its neighbors. The nodes can perform any required computation at any time. 
Each node $u$ has a unique identifier and is {\em churned in} at some round
$r_i$ and {\em churned out} at some $r_o > r_i$. More precisely, for each node
$u$, there is a maximal range $[r_i, r_o -1]$ such that for every $r \in [r_i, r_o -1]$,  $u \in V^{r}$ and
for every $r \notin [r_i, r_o-1]$, $u \notin V^r$. Any information about the
network at large is only learned through
the messages that $u$ receives. It has no a priori knowledge about who its neighbors will be in
the future. Neither does $u$ know when (or whether) it will be churned out. 
For all $r$, we assume that $|V^r| = n$, where $n$ is a suitably
large positive integer.  This assumption simplifies our analysis. Our
 algorithms can be adapted to work correctly as long as the number of nodes is
reasonably stable. Also, we assume that $n$ (or a constant factor
estimate of $n$) is common knowledge among the nodes in
the network.

\iffalse
\item[Churn:]  For each $r>1$, 
%
  $$|V^r \setminus V^{r-1}| = |V^{r-1} \setminus V^{r}| \le \L = \eps c(n)
  \text{,}$$
%
where $\L$ is the {\em churn limit}, which is some fixed $\eps>0$ fraction of the
\emph{order of the churn} $c(n)$; the equality in the above equation
%
ensures that the network size remains stable.  Our work is aimed at high levels
of churn up to a churn limit $c(n) = n/\polylog(n)$. 
%
\item[Bounded Degree Expanders:] The sequence of graphs $(G^r)_{r\geq 0}$ is an
  expander family with a vertex expansion of at least $\alpha$. 
  %
  %
  %
  %
\fi  

%
\onlyShort{
\noindent {\bf The Storage and Search Problem.}}
\onlyLong{\subsection{The Storage and Search Problem}}
In simple terms, we want to build a robust distributed solution for the storage and retrieval of data items.
Nodes can produce data items.
Each data item is uniquely identified by an id (such as its hash value). 
When a node produces a data item, the network must be able to place and maintain copies of the data item in several nodes of the network. 
To ensure scalability, we want to upper bound the number of copies of each
data item, but more importantly, we must also replicate the data
sufficiently to ensure that, with high probability, the churn does not destroy all copies of a data item.
When a node $u$ requires  a data item (whose id, we assume, is known to the node),  it must be able to access the data item within a bounded amount of time. 
To keep things simple, we only require that $u$ knows the id of {\em a} node (currently in the network) that has the data item $u$ needs.
We ideally want an arbitrarily large number of data items to be supported.  

\section{Random Walks Under Churn}  
\label{sec:soup}
As a building block for our solution to the storage and search problem, we study some basic properties of random walks in dynamic networks with churn. It is well-known that random walks on expander graphs exhibit fast mixing time, thus allowing near uniform sampling of nodes in the network. This behavior quite easily extends to expander networks in which edges change dynamically, but nodes are fixed~\cite{SMP12}. It is more challenging to obtain such characteristics under networks under adversarial node churn. One issue is that random walks may not survive. The more challenging issue is that adversarial churn may bias the random walks, which in turn will bias sampling of nodes. We address both issues in our analysis. 
In particular, we show that for any time $t$, most of the random walks that were generated at time $t$ survive up to time $t+O(\log n)$ and at that time the surviving walks are close to uniformly distributed among the existing nodes. 
\onlyLong{
We also show on the flip side that the origin of a random walk that survived in the network for $O(\log n)$ rounds is uniformly distributed among nodes that existed at the time of its origin.  We assume the churn rate to be at most   $4n/\log^k n$; we show that $k$ can be of the form $1+\delta$ for any fixed  $\delta>0$. We begin by defining some terms and establishing some notations. }

Let $\frac{1}{n}{\bold 1}$ be the uniform distribution vector that assigns  probability $1/n$ to each of the $n$ nodes in $G^t$. We define $\pi(\G, s, t, t_0)$, $t \ge t_0 \ge 0$, to be the probability distribution vector of the position of a random walk in round $t$,  given that the random walk started at $s \in G^{t_0}$ in round $t_0$ and proceeded to walk in the dynamic network $\G$. For our purposes, we will restrict our attention to random walks that start in round $0$, so, for convenience, we use $\pi(\G,s,t)$ to refer to $\pi(\G,s,t,0)$. The  component $\pi_d(\G,s,t)$ refers to the probability that the random walk is at $d \in V^t$ in round $t$. Since the random walk could have been terminated because of churn,  we use  $\pi_*(\G,s,t) = 1 - \sum_{d \in V^t}\pi_d(\G,s,t)$ to denote the probability that the random walk did not survive until round $t$. We are now ready to present a key ingredient in our algorithms, namely the Soup Theorem, which may also be of independent interest in dynamic graphs (and not just P2P \footnote{In particular, the Soup theorem applies to any expander topology model with churn. It can also
be extended to general connected networks, although the bounds will depend on the dynamic mixing time \cite{SMP12}  of the underlying dynamic network.})  with churn. 
\begin{theorem}[Soup Theorem] \label{thm:soup}
Suppose that the churn is limited by $4n/\log^k n$, where $k=1+\delta$ for any 
fixed $\delta > 0$.
With high probability, there exists a set of nodes $\core \subseteq V^0 \cap 
V^{2 \tau}$,  with cardinality at least $n - 
\frac{8n}{\log^{(k-1)/2} n}$ such that for any $s \in \core$ and $d \in \core$, 
a random walk that starts in $s$ terminates in $d$ (in round 
$2\tau$) with probability in $[1/17n, 3/2n]$.
\end{theorem}
\onlyShort{We present a high level sketch of the proof and defer the details to the full version of the paper.
\begin{proof}[sketch] With near linear churn per round, we cannot assume that a 
  random walk will survive in the network, let alone distribute evenly 
  throughout the network.  Therefore, for analysis purposes, we construct a 
  dynamic graph $\bG$ that mimics $\G$, but with the (artificially) added 
  advantage that  random walk in $\bG$ survive with probability 1. When a node 
  $v \in \G^i$ is churned out, all the random walks at $v$ are killed. Recall 
  that we have assumed that the number of nodes churned in at any round equals 
  the number of nodes churned out. To obtain $\bG$, for each $v$ that is churned 
  out, we pick a unique node $v' \in \G^{i+1}$ that was churned in at round 
  $i+1$ and place all the random walks previously at $v$ on $v'$.  The dynamic 
  network $\bG$ thus obtained satisfies  the type of network studied 
  in~\cite{SMP12}, i.e., the edges are rewired arbitrarily but the nodes are 
  unaffected. The dynamic mixing time of a random walk starting from some 
  initial distribution over the nodes of a network at time step 0 is the time it 
  takes for the random walk to be distributed nearly uniformly over the nodes of 
  the network. More formally, the dynamic mixing time
\[
T(\bG,\frac{1}{2n}) \triangleq \min \left \{ t : \max_s || \pi(\bG, s, t)
- \frac{1}{n}{\bold 1}||_\infty \le \frac{1}{2n} \right \}.
\]
In~\cite{SMP12}, Sarma {\it et al.} show that  $T(\bG, 1/2n) \in O\left ((\log 
  n)/(1 - \lambda)\right ),$
where $\lambda$ is an upper bound on the second largest eigenvalue in absolute 
value for all $\bar{G} \in \bG$.
Moreover, the expected number of random walks that are at any specific node at 
any point in time is the same as the initial token distribution.
Thus, with high probability, all tokens are able to take $T$ steps within  
$\tau \in O(\log n)$ rounds.
For the rest of the analysis, we use this fast mixing behavior of random walks 
in $\bG$ to make inferences on the behavior of random walks in $\G$.
The sequence of inferences we make are as follows.

\noindent 1. Given a dynamic graph process $\G$, we show that there is a large 
set of nodes  $S \subseteq V^0$ of cardinality at least $n-4n/\log^{(k-1)/2} n$ 
such that every  random walk generated in each of these nodes at time 0 survives 
up to the mixing time with probability $1-1/\log^{(k-1)/2}n$.  \item We then 
show that for every $s \in S$ there is a set $D(s) \in V^{\tau}$ also of 
cardinality at least $n - 4n/ \log^{(k-1)/2}n$  such that a random walk that 
starts from $s$ at time step 0 is almost uniformly distributed in $D(s)$.  More 
precisely, for any $d \in D(s)$, $1/4n \le \pi_d(\G,s,\tau) \le 
3/2n$. 

\noindent 2. Taking advantage of the reversibility of random walks in $\bG$, we 
show that there is a set $D \subseteq V^{\tau}$ such that for every $d \in D$, 
there is a set $S(d) \subseteq V^0$ again of cardinality at least 
$n-4n/\log^{(k-1)/2} n$ such that the origin of every  random walk that 
terminated in $d$ is almost uniformly distributed in $S(d)$, i.e., for any 
specific $s \in S(d)$, the random walk the probability that the random walk 
originated from $s$ lies in the range $[1/4n, 3/2n]$.

\noindent 3. Carefully combining two $\tau$ round phases, we show that there is 
a large set of nodes $\core \subseteq V^0 \cap V^{2 \tau}$ of cardinality at 
least  $n-8n/\log^{(k-1)/2} n$ such that, for any fixed pair $s, d \in \core$, a 
random walk that starts from $s$ in round 0 will reach $d$ in round $2\tau$ with 
probability in $\Theta(1/n)$ and likewise a random walk  that started in round 0 
and terminated at $d$ in round $2 \tau$ originated in $s$ with probability in 
$\Theta(1/n)$, thus proving the theorem. 
\end{proof}
The upshot of Theorem\ \ref{thm:soup} is that for any fixed period of $2 \tau$  
rounds we are guaranteed a large set \core\ of cardinality at least $n - o(n)$ 
such that random walks starting and ending in \core\ are well mixed.
}
\onlyLong{
Given a dynamic network $\G$, for the purpose of analysing random walks, we  construct a corresponding random walk preserving dynamic network $\bar{\mathcal{G}}$ as follows. 
We initialize the network by setting $\bar{G}^0$ to $G^0$. Each  $\bar{G}^t \in \bG$, $t > 0$, is essentially the same graph as $G^t$ except that we copy the state of  each node that is churned out in round $t$ on to a unique node that is churned in at round $t$. 
Intuitively, $\bG$ is a testbed network corresponding to $\G$ in which, random walks that are eliminated in $\G$ are preserved.
Notice that there is a one to one correspondence between objects (vertices, edges, and graphs) in $\G$ to the objects in $\bG$.
For notational clarity, we refer to the sequence of graphs in $\bG$ as $(\bg^0, \bg^1, \ldots)$, where each $\bg^t = (\bv^t, \be^t)$. Given a vertex $v \in V^t$ (resp. $e \in E^t$, $S \subseteq V^t$, etc), we denote the corresponding vertex in $\bv^t$ by $\bar{v}$ (resp., $\bar{e} \in \be^t$, $\bar{S} \subseteq \bv^t$, etc). 
Let $\bG$ be the random walk preserving dynamic network constructed from a  $\d$-regular non-bipartite dynamic network $\G$.  We get the following characterizations of random walks.

\begin{lemma} \label{lem:mt}
Suppose that an oblivious adversary fixes the dynamic network $\G$
of $\d$-regular non-bipartite expander graphs in advance and let $\bG$ be the corresponding
preserving dynamic network.
If each correct node $s$ starts $h\log n$ random walks of length $T\in
\Theta(\log n)$ and forwards up to $2h\log n$ random walk tokens per round, then there is a fixed $\tau =\m\log n\in \Theta(\log n)$, for
some constant $\m>0$, s.t.\  the following hold:
\begin{compactenum}
\item[\bf (A)] Every walk in $\bG$ completes $T$ steps in $\tau$ rounds w.h.p.
\item[\bf (B)] Every walk in $\bG$ that starts at some node $s$  has 
  probability in $\left[\frac{1}{2n},\frac{3}{2n}\right]$ of being at any node 
  $d$ after taking $T$ steps.  Formally, $\forall s\in \bv^0\ \forall d\in 
  \bv^T\colon
  \frac{1}{2n}\le \pi_d(\bG,s,\tau) \le \frac{3}{2n}$.
\end{compactenum}
\end{lemma}
\begin{proof}
For any $c>0$, we know by \cite{SMP12} that after taking $T=T(c) \in \Theta(\log n)$ steps in
a dynamic $\d$-regular expander network with a changing edge topology, a random walk
token has probability 
$[\frac{1}{n}-\frac{1}{n^{1+c}},\frac{1}{n}+\frac{1}{n^{1+c}}]$ of
being at any particular node in step $T$, which shows (B).

Recalling that all graphs $\bG$ are $\d$-regular and that initially each node
generates $h\log n$ tokens, the expected number of tokens received by a node is
$h\log n$, in any round $r$.
Applying a standard Chernoff bound, it follows that with probability $\ge
1-n^{-4}$, each  node receives at most $2h\log n$ tokens in
$r$.
Taking a union bound over $2h\log n$ rounds and all  nodes, the same is
true at every correct node during rounds $[1,2h\log n]$ with probability $\ge
1-n^{-2}$.
Since each node can forward up to $2h\log n$ tokens per round, we set $\m = 2h$.
This guarantees that (w.h.p.) every token is forwarded once in every round.
Thus, with high probability, all walks complete all $T$ steps in 
$\tau=\m\log n$ rounds and satisfy the required probability bound.
\qed
\end{proof}

We call $\tau$ the \emph{dynamic mixing time} of a random walk.
It follows therefore that when every $\bg \in \bG$ is an expander with a second largest eigenvalue bounded by $\lambda$ which is is a fixed constant bounded away from 1, as we have assumed,  then the mixing time $\tau (\bG,\frac{1}{2n})  = \m \log n$, where $\m$ is a fixed constant known to all nodes in the network. In particular, 
for every $ s\in \bv^0$ and $d \in \bv^{\tau}$, 
$\frac{1}{2n} \le \pi_d(\bG, s,\tau) \le \frac{3}{2n}.$

We first show that given a dynamic graph process $\G$, there is a large set of nodes at time 0, such that a random walk generated in each of these nodes at time 0 survives up to the mixing time with probability $1-1/\log^{(k-1)/2} n$.

\begin{lemma}\label{lem:alive}
Consider a churn of $4n/\log^k n$ and let $S \triangleq \{s : (s \in V^0) \wedge (\pi_*(\G,s,\tau) \le 1/\log^{(k-1)/2} n)\}$. 
Then, $|S| \ge n - 4 n/\log^{(k-1)/2} n$.
\end{lemma}
\onlyLong{
\begin{proof}
Start one random walk from each node in $V^0$.
Each $\bg \in \bG$ being $\d$-regular, the expected number of random walks on any node in $\bG$ at any round is 1. 
Therefore, in $\G$, after $\tau \in O(\log n)$ rounds, the expected number of random walks that will be eliminated is 
\begin{equation}\label{eqn:numKilled}
\sum_{s \in V^0} \pi_*(\G,s,\tau) \le \frac{4 n}{ \log^{k-1} n},
\end{equation}
for some suitable constant $c$.
Since $\pi_*(\G,s,\tau) \ge 1/\log^{(k-1)/2} n$ for every  $s \in V^0 \setminus S$, we get 
\[\frac{1}{\log^{(k-1)/2} n} |V^0 \setminus S| \le \frac{4 n}{ \log^{k-1} n}.\] 
This implies that $|V^0 \setminus S| \le 4 n/\log^{(k-1)/2} n$.
Since $|V^0| = n$, the lemma follows.
\end{proof}
}%

Consider a random walk that started at time 0 from any $s\in S$. We now endeavor
to show that if the random walk survives for the dynamic mixing time $\tau$ then its destination will be close to uniformly distributed.
\begin{lemma}\label{lem:mix}
Suppose that we have a  churn limit of $4n/\log^k n$, where $k = 1+\delta$ for 
any fixed $\delta > 0$.
With high probability, there exists a set $S \subseteq V^0$ (as defined in 
Lemma~\ref{lem:alive}) of cardinality at least $n - \frac{4n}{\log^{(k-1)/2} 
  n}$ with the property  that for every $s \in S$, there exists a $D(s) 
\subseteq V^{\tau}$ of cardinality at least $n - 
\frac{4n}{\log^{(k-1)/2} n}$ such that for every $d \in D(s)$, $1/4n \le 
\pi_d(\G,s,\tau) \le 3/2n$.
\end{lemma}
\onlyLong{
\begin{proof}
A random walk that survives in $\G$ also survives in $\bG$, thus for every $s \in V^0$ and $d \in V^{\tau}$, 
\[
\pi_d(\G,s, \tau) \le \pi_d(\bG,s, \tau) \le 3/2n.
\]
The crux of the proof is showing that a random walk from $s\in S$ reaches a large number of nodes in $V^{\tau}$, each with probability at least $\frac{1}{4n}$.
From Lemma\ \ref{lem:alive}, for every $s \in S$, $\pi_*(\G,s,\tau) \le 1/\log^{(k-1)/2} n$.  Consider a walk that started at time 0 at $s\in S$. Summing over 
all the possible locations of the walk at time $\tau$ we have  
\begin{equation}\label{eqn:summation}
\sum_{d \in V^{\tau}} (\pi_d(\bG,s, \tau) - \pi_d(\G,s, \tau)) = \pi_*(\G,s,\tau) \le 1/\log^{(k-1)/2} n.
\end{equation}
To lower bound $|D(s)|$, we upper bound the cardinality of the complement set \[\hat{D} \triangleq V^{\tau} \setminus D(s) = \{ d : (d \in  V^{\tau}) \wedge (\pi_d(\G, s, \tau))  < 1/4n)\}.\] Since $\hat{D} \subset V^{\tau}$, we can restrict the summation in Equation~\ref{eqn:summation} to the elements of $\hat{D}$ and get 
\begin{equation}
\label{fb}
\sum_{d \in \hat{D}} (\pi_d(\bG,s, \tau) - \pi_d(\G,s, \tau) ) \le 1/\log^{(k-1)/2} n.
\end{equation}
By Lemma~\ref{lem:mt}.(A), with high probability all walks have completed $\tau$
steps with high probability, and, by Lemma~\ref{lem:mt}.(B), $\pi_d(\bG,s, \tau)
\ge 1/2n$ for every $d \in  \bar{V}^{\tau}$, but $\pi_d(\G,s, \tau) \le 1/4n$ 
for every $d \in \hat{D}$.
Thus, for any $s\in S$ and $d\in \hat{D}$, $$\pi_d(\bG,s, \tau) - \pi_d(\G,s,
\tau)\geq \frac{1}{4n}.$$
Plugging to equation~(\ref{fb}) we have
$
|\hat{D}| \le (1/\log^{(k-1)/2} n)/(1/4n) = 4n/\log^{(k-1)/2} n
$,
thus establishing the lemma.
\end{proof}
}%

In Lemma~\ref{lem:mix}, we studied the distribution of the destination of random walks originating from a large set $S \subset V^0$. Similarly, in Lemma~\ref{cor:reverse}, we formalize our understanding of the origin of  random walks that terminate in some large set $D\subseteq V^{\tau}$. 
\onlyShort{The proof of the following lemma follows quite immediately; see full version for complete proof.}
\begin{lemma}[Reversibility of Random Walks] \label{cor:reverse}
Suppose that the churn is limited by $4n/\log^k n$.
With high probability, there exists a set $D \subseteq V^{\tau}$ of 
cardinality at least $n - \frac{4n}{\log^{(k-1)/2} n}$ such that, given a node 
$d \in D$, there is a set $S(d) \subseteq V^0$ of cardinality at least $n - 
\frac{4n}{\log^{(k-1)/2} n}$ such that a random walk that terminated in $d$ 
originated in any $s \in S(d)$ with probability in the range $[1/4n, 3/2n]$.
\end{lemma}
\onlyLong{
\begin{proof}
Notice first that the reverse sequence of graphs $\overleftarrow{\G} = 
(G^{\tau}, G^{\tau -1}, \ldots, G^0)$ is a valid sequence of graphs that 
make up a dynamic network, albeit one that has a limited number of rounds to 
offer.  Furthermore, let $u$ and $v$ be two neighbours in $\G$ at some time $t$. 
Since $\bg^t$ is $\d$-regular, the probability of a random walk on $u$ moving to 
$v$ equals the probability of a random walk moving from $v$ to $u$. Therefore, 
to study the distribution of the origin of a random walk (in $V^0$) that 
terminated in some node $t$ in $G^{\tau}$, we can initiate a 
random walk in the same  node $t$ in $\overleftarrow{\G}$ (in round 0) and study 
the distribution of the random walk's destination in $G^0$ (in round 
$\tau$).
Lemma\ \ref{lem:mix} applies to $\overleftarrow{\G}$ implying that (w.h.p.) 
there exists  sets $D \subseteq V^{\tau}$ and $S \subseteq 
V^0$, both of cardinality at least $n -\frac{4n}{\log^{(k-1)/2} n}$,  such that 
a  random walk originating  in some $d\in D$ in round 0 of  
$\overleftarrow{\G}$ terminated   in some $s \in S$ in round 
${\tau}$. The lemma follows when we view this random walk 
property from the perspective of $\G$.
\end{proof}
}
Combining Lemmata\ \ref{lem:mix} and \ref{cor:reverse} carefully we can complete the proof of Theorem~\ref{thm:soup}.

\onlyLong{
\begin{proof}[of Theorem \ref{thm:soup}]
The upper bound of $3/2n$ on the probability that $s$ terminates in $d$ follows
quite easily from Lemma~\ref{lem:mix} when we note that the probability with
which a random walk terminates at any node does not increase over time.
Therefore, we focus on the lower bound.

First, we choose $D \subseteq V^{2 \tau}$ based on Lemma~\ref{cor:reverse} such 
that (w.h.p.) for any $d \in D$, there is a set $\overleftarrow{D}(d) \subseteq 
V^{ \tau}$ of cardinality at least $n - \frac{4n}{\log^{(k-1)/2} n}$ such that, 
for every $d' \in \overleftarrow{D}(d)$, a random walk that terminated in $d$ 
originated from $d'$ with probability at least $1/4n$.

We then choose $S \subseteq V^0$ based on Lemma~\ref{lem:mix} such that (w.h.p.) 
for any $s  \in S$, there is a set $\overrightarrow{S}(d) \subseteq V^{ \tau}$ 
of cardinality at least $n - \frac{4n}{\log^{(k-1)/2} n}$ such that, for every 
$s' \in \overrightarrow{S}$, a random walk that originated in $s$ terminates in 
$s'$ with  probability at least $1/4n$.

Notice that the cardinality of $\overrightarrow{S} \cap \overleftarrow{D}$ for any $s \in S$ and $d \in D$ is at least $n - \frac{8n}{\log^{(k-1)/2} n}$. 

We now fix a pair $(s,d)$, where $s \in S$ and $d \in D$ and consider a random walk that terminated in $d$. The random walk was in some node in $\overrightarrow{S}  \cap \overleftarrow{D}$ in round $\tau$ with probability at least  
\[\sum_{x \in \overrightarrow{S}  \cap \overleftarrow{D}} \frac{1}{4n} = (1/4) - o(1) \ge (1/4) - \hat{\varepsilon},
\] 
for any fixed $\hat{\varepsilon} > 0$.  Let us now condition our random walk on the event that the random walk was on some $x \in \overrightarrow{S}  \cap \overleftarrow{D}$ in round $\tau$. Then, it originated from $s$ with probability $1/4n$. Therefore, the random walk that terminated in $d$ originated from $s$ with probability 
$
(1/4 - \hat{\varepsilon} ) (1/4n) \ge (1/17n)
$ when $\hat{\varepsilon} \le 1/68$. The theorem follows by setting $\core 
\triangleq S \cap D$ because $|\core|$ is also at least $n - 
\frac{8n}{\log^{(k-1)/2} n}$.
\end{proof}
}%

}

%
\begin{comment}
\begin{corollary}
  \label{cor:numberOfWalks}
Suppose each node in $V^0$ initiates  $\alpha \log n$ random walks. Then, after ${2 \tau} \in \Theta(\log n)$ rounds, the expected number of random walks at each node in the set $\core \subseteq V^0 \cap V^{2 \tau}$ as defined in Theorem\ \ref{thm:soup} (in  round $2 \tau$) will be in $[\frac{\alpha}{17} \log n, \frac{3\alpha}{2} \log n]$. In fact, with probability at least $1 - 1/n$, every node in $\core$ will have at least $\frac{\alpha-\sqrt{68\alpha}}{17}\log n$ random walks.
\end{corollary}
\end{comment}

%
%
%
\section{Storage and Search of Data}\label{sec:storage}
In this section we describe a mechanism that enables all but $o(n)$ nodes 
to persistently store data in the network.
We will assume that churn rate is $4n/\log^k n$. 
A key goal is  to tolerate as much churn as possible, hence we would like $k$ to be as small as possible. With this in mind, we again show in the analysis that $k$ can be of the form $1+\delta$ for any fixed $\delta > 0$.

A na\"{\i}ve solution is to flood the data through the network and store it at 
a linear number of nodes, which guarantees fast retrieval and persistence with 
probability $1$.
Clearly, such an approach does not scale to large peer-to-peer networks 
due to the congestion caused by flooding and the costs of storing the item 
at almost every node.
As we strive to design algorithms that are useful in large scale 
P2P-networks, we limit the amount of communication by using random walks 
instead of flooding and require only a sublinear number of nodes to be 
responsible for the storage of an item --- only $\Theta(\log n)$ of these 
nodes will actually store the item\footnote{In fact (as we noted earlier) using erasure coding techniques,
the overall storage can be limited to a constant factor overhead; see \onlyShort{the full version of the paper.}\onlyLong{Section~\ref{sec:erasure-codes}.}} whereas the other nodes serve as 
landmarks pointing to these $\Theta(\log n)$ nodes.

Suppose that node $u$ wants to store item $\I$ and assume that $u$ is part 
of the large set of nodes $\Core$ provided by Theorem~\ref{thm:soup}, 
which consists of nodes that are able to obtain (almost uniform) node id 
samples from the same set, despite churn.
A well known solution is to make use of the birthday paradox:  If node $u$ 
is able to select $\Theta(\sqrt{n}\log n)$ sample ids and assign these so 
called \emph{data nodes} to store $\I$, then $\I$ can be retrieved within 
$\sqrt{n}$ rounds by most nodes, with high probability.
In our dynamic setting,  up to $O(n/\log^k n)$ nodes per round can be affected 
by churn, which means that the number of data nodes might decrease rapidly.
Care must be taken when replenishing the number of data nodes, as we need 
to ensure that the data nodes are chosen randomly and their total number does not exceed $\tilde{O}(\sqrt{n})$.
A simple algorithm for estimating the actual number of data nodes is to 
require data nodes first to generate a random value from the exponential 
distribution with rate $1$, then to aggregate the minimum generated value 
$z$ by flooding it through the network (cf.\ \cite{APRU12}), and finally 
to compute the estimate as $1/z$.
The simplicity of the above approach comes at the price of requiring every 
node to participate (by flooding) in the storage of the item.

We now describe an approach that avoids the above pitfalls and
provides fast data retrieval and persistence with high probability, while
limiting the actual number of nodes needed for storing a data item to
$\Theta(\log n)$, while a large set of $\Omega(\sqrt{n})$ nodes serve as 
so-called ``landmarks''. 
That is, a node $v$ is a \emph{landmark for item $\I$ in
$r$}, if $v$ knows the id of some node $w \in V^r$ that stores $\I$.
Note that even if $v$ was a landmark in $r$, it might no longer be a
landmark in round $r+1$ if $w$ has been churned out at the beginning of
$r$; moreover, $v$ itself will not be aware of this change until it
attempts to contact $w$.
To facilitate the maintenance of a large set of randomly distributed
landmarks, our algorithms construct a committee of $\Theta(\log n)$
nodes via the overlay network.
In the context of the storage procedure, the committee is responsible for
storing some data item $\I$ and creating sufficiently many (i.e.\
$\Omega(\sqrt{n})$) randomly distributed \emph{storage landmarks} for
allowing fast retrieval of $\I$ by other nodes.
If, on the other hand, $u$ wants to retrieve item $\I$, having a large
number of \emph{search landmarks} will significantly increase the probability of
finding a sample of a storage landmark in short time.
Due to churn, the number of landmark nodes (and the number of committee
members) might be decreasing rapidly.
Thus the committee members continuously need to replenish the committee
and rebuild the landmark set.
Note that we guarantee that the number of landmarks involved with a
storage or search request remains in $\tilde{O}(n^{1/2+\delta})$, for any constant
$\delta>0$, which ensures that our algorithms are scalable to large networks.

\subsection{Building Block: Electing and Maintaining a Committee} 
\label{sec:committee}
We will now study how a node $u$ can elect and maintain a committee of 
nodes in the network. Such a committee can be entrusted with some task 
that might need to be performed persistently in the network even after $u$ 
is churned out.  We for instance use such a committee in Section\ \ref{subsec:sr}  to enable $u$ to store a 
data item $\I$ so that some other node that needs the data may be able to 
access it well into the future without relying on $u$'s presence in the 
network.
While electing a committee is easy, we need to be careful to maintain the committee for a longer (polynomial in $n$) period of time because, without maintenance, 
the members can be churned out in $O(\log^k n)$ rounds.

\onlyLong{
In Section\ \ref{sec:soup} we analyzed the setting where each node in the dynamic network initiates $\Theta(\log n)$ random walks in round 1. 
We focused on a time span of
$\Theta(\log n)$ rounds to study the mixing characteristics of random walks culminating in Theorem\ \ref{thm:storage}. 
Informally speaking, we showed that a large number of nodes in round $\tau =  \tau(\bG, 1/2n)$ received good samples of nodes currently in the network. This will help us create a committee of randomly chosen nodes, but churn can decimate this committee in $O(\log^k n)$ rounds.  
Our goal now, however, is to maintain a committee for a longer period of time --- time that is polynomial in $n$.} 
Towards this goal, we make each node initiate $\alpha \log n$ random walks every round. Depending on how long we want the committee to last, we can fix an appropriately large $\alpha$. Each random walk travels for $2\tau$ rounds; the node at which the random walk stops  is called its destination. The destination node can use the source  of the random walk as a sample from the set of nodes in the network. Since every node initiates some $\alpha \log n$ random walks every round, Theorem\ \ref{thm:soup} is applicable in every round $r \ge 2\tau$.
To formalize this application of Theorem\ \ref{thm:soup}, we parameterize $\core$ with respect to time.  
We define $\core^r$ to be the largest   subset of  $V^{r - 2\tau} \cap 
V^{r}$  such that for any $s \in \core^r$ and $d \in \core^r$, a random 
walk that starts from $s$ (in round $r - 2\tau$) terminates in $d$ (in 
round $r$) with probability in $[1/17n, 3/2n]$. From Theorem\ 
\ref{thm:soup}, we know that $\core^r$ has cardinality at least $n - 
O({n}/{\log^{(k-1)/2} n})$. When the value of 
$r$ is clear from the context, we may avoid the explicit superscript.

Algorithm~\ref{algo:committee} presents an algorithm that
\begin{enumerate}[noitemsep,topsep=0pt,parsep=0pt,partopsep=0pt]
\item enables a node $u \in \core^r$, $r \ge 2\tau$,  to elect a committee of $\Theta(\log n)$ 
nodes and
\item enables the  committee to maintain itself at  a cardinality of 
  $\Theta(\log n)$ nodes despite $O(n/\log^k n)$ churn. Moreover, the 
  committee must comprise of at least $\Theta(\log n)$ nodes from the current 
  $\core$.
\end{enumerate}
In Algorithm~\ref{algo:committee}, we assume that $u$ is in the $\Core$ when it needs to 
create the committee. We  show \onlyLong{in Lemma\ \ref{lem:sufficient}} that, 
if $u \in \Core$, then, $u$ will receive a sufficient number of random 
samples, so it chooses some $h \log n$ samples to form the committee. 
To ensure that churn does not decimate the committee, every $\Theta(\log n)$ 
rounds, we re-form the committee, i.e., the current committee members choose 
a suitable leader \onlyLong{(denoted $c^r$ in Algorithm\ \ref{algo:committee})}
that chooses a new set of committee members. The old committee members hand 
over their task to the new committee members and ``resign" from the committee 
and the new members join the committee and resume the task they are called to perform.

\begin{algorithm}[h!]
  \begin{algorithmic}[1]
\item[]

 ~
 
\noindent {\bf \underline{Committee Creation}.}

~

\item[] \Comment{Let $r_1  \ge  2\tau$ be the round when $u$ must create $\Com$. We assume that $u \in \Core^{r_1}$. Let $h \le \alpha/36$ be a fixed constant.}
 \item[{\bf At  round $r_1$: }] Node $u$ chooses $h\log n$ sample ids
 and requests each 
    node to join the committee $\Comm$. Therefore,
    $
    \Comm \leftarrow \set{v | (v \in V^{r_1}) \wedge \text{$v$ received an invitation from $u$}}.
    $
  Along with the request, $u$ sends all the ids
    in $\Comm$ to every node in $\Comm$.
    This enables  the nodes in $\Comm$ to form a clique interconnection.

~

\noindent {\bf \underline{Committee Maintenance}.}

  \item[]
     \item[] \Comment{For every round $r$ that is  $2\gamma \tau$ rounds after $\Com$ is created for every  positive integer $\gamma$.} 
     \item[\bf At round $r$:]The nodes in $\Comm$ record the random walks they receive along with the source of each random walk.
          \item[{\bf At round $r+1$: }]The nodes in $\Comm$ exchange the number of random walks they received in round $r$ with each other. 
     \item[{\bf At the end of round $r+1$: }] The number of random walks received by each node in $\Comm$ is common knowledge among the members of $\Comm$.  The node $c^r$ with the largest number of random walks is chosen to initiate the new committee (breaking ties arbitrarily yet unanimously). The choice of $c^r$ is now common knowledge among the nodes in $\Comm$. 
     \item[{\bf At round $r+2$: }] The node $c^r$ chooses $h\log n$ random walks that stopped at $c^r$ in round $r$ and invites{\footnotemark} their source nodes  to form the new committee in round $r+3$. Let $\Comm^*$ be the set of invited source nodes. Along with the invitation, the $h \log n$ id's of all members of $\Comm^{*}$ are included. Therefore, the id's of nodes in $\Comm^{*}$ becomes common knowledge among the nodes in $\Comm^{*}$.  
     The nodes in $\Comm$ cease to be members of the committee at the end of round $r+2$. (If the situation calls for it, we may postpone the ``resignation" of the current committee members; the overlap in membership can be used for ensuring smooth transition of the task performed by the committee.)
     \item[{\bf At round $r+3$: }] The members in $\Comm^{*}$ formally take over the committee. I.e. $\Comm \leftarrow \Comm^{*}.$
      
      Each member of the new $\Comm$ uses the id's of all other members to form a clique interconnection. 

  \end{algorithmic}
  \caption{Committee Maintenance and Construction for node $u$.}
  \label{algo:committee}
\end{algorithm}

\footnotetext{In our algorithm description, we assume for simplicity that $c^r$ is not
churned out in round $r+2$.  We can handle the case where $c^r$ \emph{is}
churned out, by having the set $S$ of the $\Theta(\log n)$ committee
members that have received the largest number of random walks all perform
the task of $c^r$ in parallel, i.e., each of them builds a new committee.
Once these committee constructions are complete, the (survived) nodes in
$S$ agree on a single member $c^*$ of $S$ and its committee $\Com^*$, and
all other committees are dissolved.}

\onlyLong{

We begin our analysis of Algorithm\ \ref{algo:committee} with a lemma that limits the number of random walks that any node receives in round $r$ from nodes that are not currently in the $\core$.

\begin{lemma}\label{lem:bad_rand_walks}
Consider a churn rate of $4n/ \log^k n$.
For any $r \ge 2 \tau$ and any $u \in V^r$, let $B(u,r)$ be the number of random walks that  started in round $r - 2 \tau$ from some node in $V^{r - 2 \tau} \setminus \Core^r$  and stopped in $u$ in round $r$. Then, $E[B(u,r)] \le 6 \alpha \log^{-\frac{3-k}{2}} n$ and %
$\Prob{B(u,r) \ge 12 \alpha \log^{\frac{5 - k}{4}} n} 
\le 1/n^{2 \alpha}.$
\end{lemma}
\onlyLong{
\begin{proof}
Recall that we analyzed random walks in $\G$ using $\bG$. Consider node $\bar{u}
\in \bG$ that corresponds to $u \in \G$. Let $v$ be some node in  $V^{r-2 \tau}
\setminus \Core^r$ and let $\bar{v}$ be the corresponding node in $\bG$. From
Lemma~\ref{lem:mt}.(B), we know that a random walk that started in $\bar{v}$ will reach $\bar{u}$ with probability at most $3/2n$. Therefore,
\begin{align*}
E[B(u,r)] &\le E[\text{number of random walks from $V^{r - 2\tau} \setminus \Core^r$ that reach $u$}]\\
&\le \left (\frac{4n}{\log^{\frac{k-1}{2}} n} \right ) \frac{3}{2n} \alpha \log n
= \frac{6 \alpha \log n}{\log^{\frac{k-1}{2}}n} = 6 \alpha \log^{-\frac{3-k}{2}} n.\\
\end{align*}

Now using Chernoff bounds,
\begin{align*}
\Prob{B(u,r) \ge 12 \alpha \log^{\frac{5 - k}{4}} n} & \le \Prob{B \ge (1 + \log^{\frac{k-1}{4}} n) \frac{6 \alpha}{\log^{\frac{k-3}{2}} n}}\\
			&\le \exp(-\frac{6 \alpha}{\log^{\frac{k-3}{2}} n} \cdot \frac{\log^{\frac{k-1}{2}} n}{3}) = \exp(-2 \alpha \log n) = \frac{1}{n^{2 \alpha}}.
\end{align*}
\end{proof}
}

While Lemma\ \ref{lem:bad_rand_walks} limits the number of random walks 
that a node receives from nodes not in the $\Core$, we also need to ensure 
that a node in the $\Core$ gets a sufficient number of random walks from 
other nodes in the $\core$. Thus, when we choose $c^r$, we choose a node 
that received a large number of samples. From 
Lemma~\ref{lem:bad_rand_walks}, we know that only a small number of those 
samples can be from nodes not in the $\core$, so this ensures that the 
committee that we choose will be largely from the $\core$.

\begin{lemma}\label{lem:sufficient}
Consider any $u \in \Core^r$, $r \ge 2\tau$. With probability at least $1 - 1/n^{\ell_1}$, where $\ell_1 \le \alpha / 144$, $u$ receives at least $\frac{\alpha \log n}{36}$ random walks.
\end{lemma}
\onlyLong{
\begin{proof}
Let $X$ be the number of random walks received by $u$ in round $r$. From Theorem\ \ref{thm:soup}, we know that $E[X] \ge \frac{\alpha \log n}{18}$. Using Chernoff bound, we get 
\begin{align*}
\Prob{X \le \frac{\alpha \log n}{36}}  &= \Prob{X \le (1 - 1/2) \frac{\alpha \log n}{18}} \le \exp(-\frac{\alpha \log n}{144} ) \le 1/n^{\ell_1}.
\end{align*} 
\end{proof}
}
\begin{corollary}\label{cor:sufficient}
In Algorithm\ \ref{algo:committee}, let $c^r$ be a node chosen in some round $r+1$ to invite a new set of nodes to form the committee. With probability at least $1 - 1/n^{\ell_1}$ it received more than $h \log n$ random walks in the previous round since $h \le \alpha/36$. Out of the $h \log n$ random walks, with probability at least $1 - 1/n^{2h}$, at least $ h \log n - 12 h \log^{\frac{5 - k}{4}} n$ random walks originated in $\Core^r$. 
\end{corollary}

The $h \log n$ nodes invited by $c^r$ in round $r+2$ were certainly in the network in round $r- 2\tau$. From Corollary\ \ref{cor:sufficient}, we also know that a large number of those nodes are in $\Core^r$. We also need to ensure that most of them survive for another $2\tau$ rounds until round $r + 2\t$. In particular, we want to ensure that those nodes that survive are largely in $\Core^{r+2\t}$. 

\begin{lemma}\label{lem:persist}
Suppose $I$ is the set of recipients of $h \log n$ invitations  sent by some $c^r$ in round $r+2$ (cf. Algorithm\ \ref{algo:committee}). Then with probability at most $2/n^{2h}$, $|I \setminus \Core^{r+2 \tau}| \in \omega( \log^{\frac{5-k}{4}} n)$. In other words, $|I \cap \Core^{r+2\t}| \in h\log n - o(\log n)$ whp.
\end{lemma}
\onlyLong{
\begin{proof}
From Corollary\ \ref{cor:sufficient}, we know that with probability at most $1/n^{2h}$, at most $12 h \log^{\frac{5 - k}{4}} n$ random walks did not originate in the $\Core^r$. Obviously, these $12 h \log^{\frac{5 - k}{4}} n$ random walks did not originate from nodes in $\Core^{r + 2\t}$ either.   In this lemma, we are upper bounding the cardinality of $I \setminus \Core^{r +  2\tau}$. Therefore, in addition to the $12 h \log^{\frac{5 - k}{4}} n$ samples that we have already accounted as lost, we must bound the number of samples that we lose between round $r$ to $r+ 2\tau$ due to churn. During this time period, a total of $\frac{8\t n}{\log^{k} n} = \frac{8\m n}{\log^{k-1} n}$ nodes are churned out, since $\tau = \tau(\bG, 1/2n) = \m \log n$. Each node that is churned out is an opportunity for the adversary to churn out a node in $I$. Let $X$ be the random variable that denotes the number of nodes in $I$ that were churned out between rounds $r$ to $r + 2\tau$.  Consider a node $i$ that is churned out. If $i \in I$, then the random walk from $i$ reached $c^r$, an event that can happen with probability at most $3/2n$. Therefore, whenever the adversary churns out a node, it succeeds in churning out a node in $I$ with probability at most $3/2n$. Therefore, 
$
E[X] \le (3/2n) (4\m n/\log^{k-1} n)= 6\m/\log^{k-1} n.
$ 

\begin{align*}
\Prob{X \ge 12 \sqrt{h} \log^{(2-k)/2}n} &= 
\Prob{X \ge (\frac{2\sqrt{h}}{\sqrt{\m}} \log^{\frac{k}{2}} n) \frac{6\m}{\log^{k-1} n}} \\
&\le \Prob{X \ge (1 + \frac{\sqrt{h}}{\sqrt{\m}} \log^{\frac{k}{2}} n) \frac{6\m}{\log^{k-1} n}}  &\quad\quad \text{(for sufficiently large $n$)}\\
& \le \exp(-\frac{6h \log^{k} n}{ 3\log^{k-1} n}) &\quad\quad  \text{(using Chernoff bound)}\\
&= \exp(-2h \log n) = 1/n^{2h}.
\end{align*}
Taking the union bound over the probability with which we lose  $X$ random walks plus the probability with which we lose the at most $12 h \log^{\frac{5 - k}{4}} n$ random walks that did not originate in the $\Core^r$, the result follows.
\end{proof}
}

Recall that we have assumed that  the node $u$ that seeks to create the committee in round $r_1$, $r_1 \ge 2\tau$,  is in $\Core^{r_1}$.} 
Let $\Comm^r$, $r \ge r_1$, denote the set of nodes that consider themselves to be committee members in round $r$. 
We say that $\Com^r$ is {\em good} if $|\Com^r \cap \core^r| \ge (1- \varepsilon) h \log n$ for any fixed $\varepsilon > 0$. 
In the following theorem states that the committee that is created by $u$ will be good  for a suitably long period of time.
\begin{theorem} \label{lem:committee}
Fix $\varepsilon$  to be a small positive number in $(0,1]$. Recall that $u$ creates the committee in round $r_1$. Let $R\ge r_1$ be a random variable denoting the smallest value of $r$ when $\Com^r$ is not good and let $Y$ be a geometrically distributed random variable with parameter $p = (1 /n^{\ell_1} + 2/n^{2h} ) \in n^{-\Omega(1)}$.  Then, $Y$ is smaller than $R - r_1 + 1$ in the usual stochastic order~\cite{Shaked2007}. In other words, for every positive integer $i$, $\Prob{Y \ge i} \le \Prob{R - r_1 + 1 \ge i}$.
  \end{theorem}
  \onlyShort{
  \begin{proof}[Sketch] When a new set of nodes are elected to form the committee, they perform this role for $2 \tau$ rounds and then elect a new committee. We thus show that either (i) a newly formed committee does not last  for $2 \tau$ rounds (with probability at most $1/n^{2h}$)  {\em or} is incapable of electing a new committee (with probability at most $1 /n^{\ell_1} + 1/n^{2h}$). Thus a newly formed committee fails to survive and  elect a successor committee with probability at most  $p = (1 /n^{\ell_1} + 2/n^{2h} ) \in n^{-\Omega(1)}$.
  \end{proof}
  }
  
\onlyLong{
\begin{proof}
Let $r+2$ be a round in which a new set of committee members are invited to (cf. Algorithm\ \ref{algo:committee}) by a node $c^r \in \Core^r$. 
We now show that the probability with which (i) the committee that is selected by $c^r$ is good  and (ii) remains good until $r + 2\t + 2$ (when the next set of committee members are selected) is high. The requirements of the theorem will then be subsumed.

We now list some bad events;  at least one of them must occur for the committee to cease to be good.
\begin{compactenum}
\item With probability at most $1/n^{\ell_1}$, $c^r$ will receive fewer than  $h \log n$ samples. (cf. Corollary\ \ref{cor:sufficient})
\item With probability at most $1/n^{2h}$, 
more than $12 h \log^{\frac{5 - k}{4}} n$ 
samples received by $c^r$ will not be in $\Core^r$. (cf. Corollary\ \ref{cor:sufficient})
\item With probability at most $1/n^{2h}$, 
more than $12 \m \log^{(1-k)/2}n$ nodes in $\Com^{r+2}$ will be churned out between $r+2$ and $r+2\t+2$. (cf. Lemma\ \ref{lem:persist})
\end{compactenum}

Thus, for $r+2 \le r' < r+2\t+2$, $\Com^{r'}$ will not be good with probability at most $p$. Thus, the theorem follows.
\end{proof}
}

\begin{corollary}
  \label{cor:committee}
  Let $\ell$ be a suitably large number that respects the inequality $p \le n^{-\ell}$. Suppose at some round $r+2$, a new set of committee members have been selected by  $c^r \in \Core^r$. Let $g \ge 0$ be a random variable such that $r+g+2$ is the first round after $r+2$ when the committee ceases to be good. Then, $E(g) \ge n^{\ell}$. Furthermore, for any $0 \le i \le \ell$,
$
\Prob{g \le n^{\ell - i} } \le n^{-i}.
$
\end{corollary}

\subsection{Building Block: Constructing a Set of Randomly Distributed 
  Landmarks} \label{sec:tree}
Once we have succeeded in constructing a committee of $\Theta(\log n)$ 
nodes, we can extend the ``reach'' of this committee by creating a 
randomly distributed set of nodes that know about the committee members.
An easy but inefficient solution is to simply flood the ids of the 
committee members through the network, which requires a linear number of 
messages to be sent.
In this section, we will describe a more scalable approach (cf.\ 
  Algorithm~\ref{algo:tree}) that constructs a set of $\Omega(\sqrt{n})$ 
randomly distributed nodes that know the ids of the committee members and 
thus serve as ``landmarks'' for the committee.
The basic idea is that every current committee member selects 2 of its 
received samples and adds them as children.
These child nodes in turn then attempt to select 2 child nodes each and so 
forth.
Taking into account churn, and the fact that  only $n-o(n)$ nodes are 
able to select \emph{random} child nodes, we choose a tree depth that 
ensures with high probability that the committee members will succeed to 
construct a landmark set of size at least $\Omega(\sqrt{n})$, but containing no 
more than $O(n^{1/2+\delta}\log n)$ nodes.

Due to the high amount of churn and the fact that the committee members 
change over time, the committee nodes are responsible for rebuilding the 
set of landmarks every $O(\log n)$ rounds,
which will also ensure that the landmarks are randomly distributed among 
the nodes \emph{currently} in $\Core$.
We define $\Core^{[r_1,r_2]}$ as a shorthand for
$\Core^{r_1}\cap\dots\cap\Core^{r_2}$.

\begin{algorithm}[t]
  \begin{algorithmic}[1]
  \item[] \textbf{Assumption:} There is a committee $\Com$ of $\Theta(\log n)$ nodes each of
    which is carrying out some task $\mathcal{T}$ that requires all 
    committee nodes to simultaneously start executing this algorithm.
    Task $\mathcal{T}$ can either be a retrieval of a storage request of
    some data item $\I$''. 
  \item[]
  \item[] \textbf{Every $\tau$ rounds do:}
  \STATE Every committee node $v$ tries to add $\sqrt{n}$
  randomly chosen nodes to the landmark set of $\I$ by constructing a tree:
  \STATE Node $v$ contacts its $\Theta(\log n)$ received
  sample nodes and adds $2$ nodes $v_1$ and $v_2$ that are
  not yet part of the tree as its children (if possible).
  \STATE Nodes $v_1$ and $v_2$ in turn each select $2$ (unused) nodes among
  their own samples as their children and so on.
  The nodes in the tree keep track of a tree depth counter $\mu$ that is
  initialized to $0$ and increased every time a new level is added to the
  tree. 
  The construction stops at a tree depth of
  \begin{equation} 
    \label{eq:treedepth}
    \mu = \left\lceil\frac{\log_2 n - 2\left(\log_2\log n + \log 2\right)}
          {2\log_2\left(2
          \left(1 - \frac{1}{\log^{(k-1)/2}n}\right)
          \left(1 -\frac{1}{\log^{k-1} n}\right)
          \left(1 -\frac{1}{n^{3}}\right)\right)}\right\rceil.
  \end{equation}
  Note that nodes do not need to remember the actual tree structure.
  Every time a new level of $v$'s tree is created, the parent nodes send
  all $O(\log n)$ committee ids to their respective 2 newly added children.
  \STATE Every node that has become a landmark for $\I$, remains a
  landmark for $2\tau$ rounds and then simply discards any information
  about $\I$. 
\end{algorithmic}
  \caption{Constructing a Random Set of Landmarks}
  \label{algo:tree}
\end{algorithm}

\begin{lemma} \label{lem:tree}
  Consider any round $r\ge 2\tau$ and suppose that some node $u \in 
  \Core^{r}$ executes Algorithm~\ref{algo:tree} for storing item $\I$ and
  let $T$ be the set of landmarks created for $\I$.
  Then the following holds with high probability for a polynomial number 
  of rounds starting at any round $r_1\ge r+2\tau$. For $r_2=r_1+4\tau$, 
  there exists a set $M_\I \subseteq T \cap \Core^{[r_1,r_2]}$ of 
  landmarks
  such that every node in $M_\I$ is distributed with probability in
  $[1/17n,3/2n]$ among the nodes in $\Core^{[r_1,r_2]}$ and 
  \onlyShort{$ \sqrt{n} \le |M_\I| \le |T| \le O(n^{1/2+\delta} \log 
    n).$}
  \onlyLong{
  \begin{equation} \label{eq:treesize}
    \sqrt{n} \le |M_\I| \le |T| \le O(n^{0.5+\delta} \log n),
  \end{equation}
  for any fixed constant $\delta>0$.
  }
\end{lemma}
\onlyLong{
\begin{proof}
  We will first argue the right hand side of \eqref{eq:treesize}, namely
  that the total number of landmark nodes is sublinear.
  To this end, we bound the maximal tree size of any tree created by a 
  committee member.
  For any constant $\delta>0$, there is a sufficiently large $n$, such that
  $$
  2\log_2 \left(2\left(1 - \frac{1}{\log^{(k-1)/2}n}\right)
  \left( 1-\frac{1}{\log^{k-1} n}\right)
  \left( 1-\frac{1}{n^{3}}\right)\right) \ge \frac{1}{\frac{1}{2}+\delta}.
  $$
  This allows us to bound the tree depth (cf.\ \eqref{eq:treedepth}) as
  $$
  \mu \le \frac{\log_2 n}{2\log_2\left(2
          \left(1 - \frac{1}{\log^{(k-1)/2}n}\right)
          \left(1-\frac{1}{\log^{k-1} n}\right)
          \left(1-\frac{1}{n^{3}}\right)\right)}
      \le \left(\frac{1}{2}+\delta\right)\log_2 n.
  $$
  In the worst case, all parent nodes in the tree construction always add
  $2$ child nodes, yielding a tree size of at most
  $$
  2^{(0.5+\delta)\log_2 n+1} - 1 \in O(n^{0.5+\delta}).
  $$
  Recalling that we have at most $\Theta(\log n)$ committee members,
  w.h.p., we
  get the upper bound as stated in \eqref{eq:treesize}.

  For the lower bound on $|M_\I|$, we look at the trees created by
  the committee members.
  By Corollary~\ref{cor:committee} we have that, w.h.p., 
  any node $w \in \Com^{r_1}$ receives
  $(1-\eps)h \log n \in \Theta(\log n)$ samples that are distributed with 
  probability
  $[1/17n,3/2n]$ among the nodes in $\Core^{[r_1,r_2]}$.

  Consider any parent node $v$ and assume that $v \in \Core^{[r_1,r_2]}$.
  We will bound the probability that a potential child node has already 
  been chosen as a child by some other parent node in a tree.
  By the upper bound of \eqref{eq:treesize}, we know that there are at
  most $\Theta(n^{1/2+\delta}\log n)$ nodes in the tree at any point.
  Suppose that node $v$ that has received a sample of some node $w'$ and 
  wants to add $w'$ as its child.
  Since the sampling is performed by doing independent random walks,
  the event that $w'$ has already been chosen as a child by some other 
  node (possibly in a distinct tree) is independent from $w'$ being 
  sampled by $v$.
  For sufficiently large $n$, we have
  $$
  \Prob{\text{$w'$ is already in tree} \wedge \text{w' sampled by $v$}}
  \le \frac{3n^{1/2+\delta} (1-\eps) h \log n}{2n} \le \frac{3(1-\eps)h\log 
    n}{n^{1/2-\delta}}.
  $$
  Since the parent node $v$ is in $\Core^{[r_1,r_2]}$, it follows by 
  Lemma~\ref{lem:persist} that $v$ has $h'\log n = \Theta(\log n)$ samples 
  in $\Core^{[r_1,r_2]}$ w.h.p.
  For $h'\ge 8$, the probability of $v$ not receiving at least $2$ unused 
  child nodes is at most
  $$
  \left(\frac{3(1-\eps) h \log n}{n^{1/2-\delta}}\right)^{h'\log n} \le 
  \frac{1}{n^{2\log n}} \le \frac{1}{n^3}
  $$

  Let $X_i$ be the random variable that represents the number of nodes in
  $M_\I$ up to (and including) tree level $i$;  recall that, by 
  definition, these nodes are in $\Core^{[r_1,r_2]}$.
  In addition to a fraction of $(1-\frac{1}{n^3})$ nodes that are lost due 
  to already chosen child nodes, the expectation of $X_i$ is reduced by a 
  factor of at most $(1-\frac{1}{\log^{(k-1)/2}n}$ to compensate for the 
  nodes not in $\Core^{[r_1,r_2]}$, and by the nodes that are churned out 
  during $[r_1,r_2]$, which is at most $(1-\frac{1}{\log^{k-1}})$.
  \begin{align}
  \Exp{X_i}
  & \ge 2 \Exp{X_{i-1}}
          \left(1 - \frac{1}{\log^{(k-1)/2}n}\right)
          \left(1-\frac{1}{\log^{k-1} n}\right)
          \left(1-\frac{1}{n^{3}}\right).\\
          \intertext{By Corollary~\ref{cor:committee}, we know that the 
            expected committee size (i.e.\ $\Exp{X_0}$) is at least 
            $(1-\eps)h\log n$, which shows that}
  \Exp{X_i}
  & \ge (1-\eps)h\log n \left(2\left(1 - \frac{1}{\log^{(k-1)/2}n}\right)
    \left(1-\frac{1}{\log^{k-1} n}\right)
    \left(1-\frac{1}{n^{3}}\right)\right)^{i}. \label{eq:expbound}
  \end{align}
  To lower bound the expected size of $M_\I$, we need to plug in the tree 
  height (cf.\ \eqref{eq:treedepth}) for $i$, which reveals that 
  $\Exp{X_\mu}\ge 2\sqrt{n}.$
  We use a Chernoff bound to show the lower bound on $M_\I$ as required by
  \eqref{eq:treesize}.
  From Theorem 4.5 in \cite{upfal}, it follows that
  $$
  \Prob{X \le \left(1-\frac{2\log n}{\sqrt{n}}\right)2\sqrt{n}} \le 
  \exp{-\left(-\frac{4\sqrt{n}\log n}{2\sqrt{n}}\right) } = \frac{1}{n^2},
  $$
  which proves that $M_\I \in \Omega(\sqrt{n})$ with high probability.
\end{proof}
}
\subsection{Storage and Retrieval Algorithms} \label{subsec:sr}
Now that we have general techniques for maintaining a committee of nodes 
and creating a randomly distributed set of landmarks for this committee 
(cf.\ Sections~\ref{sec:committee} and \ref{sec:tree}), we will use these 
methods to implement algorithms for storage and retrieval of data items.

\begin{definition} \label{def:available}
We say that a data item $\I$ is \emph{available in round $r$}, if the
probability of any node in $\Core^{[r,r+\tau]}$ to be in the current set of 
landmarks $M^r_\I$ is at least
$\frac{1}{\Theta(\sqrt{n})}$.
\end{definition}
It follows immediately from Corollary~\ref{cor:committee} and Lemma 
~\ref{lem:tree} that if a data item $\I$ is stored by a node $u \in 
\Core^{r_1}$ in some round $r_1$, then $\I$ will be available in the network 
for a polynomial number of rounds starting from $r_1$, with high probability.
\onlyLong{Clearly, the same is true for any later interval of polynomial number of 
rounds if $\I$ is available at its first round.}

For storing some data item $\I$ by some node $u \in \Core$, 
we combine the committee maintenance and landmark construction.
In more detail, node $u$ first creates a committee of $\Theta(\log n)$ 
nodes (cf.\ Algorithms~\ref{algo:committee}), which will be responsible for 
storing the data item, i.e., every committee member will store a copy of 
$\I$.
The committee immediately starts creating a set of $\Omega(\sqrt{n})$ 
landmark nodes, which know the ids of the committee members, but do \emph{not} 
store $\I$ itself.
Choosing these landmark nodes almost uniformly at random (cf.\ 
Lemma~\ref{lem:tree}) from the current $\Core$ set, ensures that the 
committee members can be found efficiently by the data retrieval mechanism 
described below

It follows immediately from Corollary~\ref{cor:committee} and Lemma 
~\ref{lem:tree} that if a data item $\I$ is stored by a node $u \in 
\Core^{r_1}$ in some round $r_1$, then $\I$ will be available in the 
network for a polynomial number of rounds starting from $r_1$, with high 
probability.
Owing to the memoryless nature of the persistence of the committee (cf. Theorem\ \ref{lem:committee} and Corollary\ \ref{cor:committee}), the same holds with high probability for any later interval of 
polynomial number of rounds if the data was stored in a good committee at the start of the 
interval.

\begin{algorithm}[t]
  \begin{algorithmic}[1]
  \item[] Node $u$ issues an insertion request in round $r$ for data $\I$.
  \STATE Node $u$ initiates Algorithm~\ref{algo:committee} to create a
  committee $\Com$  and requests the committee nodes to store $\I$.
  Note that the committee nodes will continue to store $\I$ on $u$'s
  behalf, even if $u$ has long been churned out.
  \STATE Moreover, $u$ instructs the committee members to execute
  Algorithm~\ref{algo:tree} and repeatedly create landmark sets of
  $\Omega(\sqrt{n})$ nodes that will respond to retrieval requests of
  $\I$.
  \end{algorithmic}
  \caption{Persistently Storing a Data Item}
  \label{algo:store}
\end{algorithm}

\begin{theorem}[Data Storage] \label{thm:storage}
  Consider any round $r \ge 2\tau$. There is a set $A$ of at least 
  $n-o(n)$ nodes, such that any data item $\I$ stored by a node in $A$ via 
    Algorithm~\ref{algo:store} in round $r$ is available for a polynomial 
  number of rounds starting from round $r+2\tau$, with high probability, in a 
  network with churn rate up to $O(n/\log^{1+\delta}n)$ per round.
\end{theorem}

Conditioning on the fact that a data item $\I$ is available in some round 
$r_i$, gives us a high probability bound that $\I$ will be available for 
another polynomial number of rounds, for any $r_i \ge r_1$. 

\begin{corollary} \label{cor:available}
  Suppose that Algorithm~\ref{algo:store} is executed for some data item 
    $\I$ since round $r_1$.
  If $\I$ is available in some round $r_i\ge r_1$, then $\I$ will be 
  available for a polynomial number of rounds starting from $r_i$ with 
  high probability.
\end{corollary}

For efficient retrieval of an available data item, we will again use the 
committee maintenance and landmark construction techniques.
To distinguish between the nodes that are serving as landmarks or 
committee members for the storage procedures from the committee and 
landmark sets that are created for data retrieval, we will call the former 
\emph{storage landmarks}, resp.\ \emph{storage committee} and the latter 
\emph{search landmarks} resp.\ \emph{search committee}.

When a node $u\in \Core^r$ executes Algorithm~\ref{algo:search} to 
retrieve some available data item $\I$, it first creates a search 
committee via Algorithm~\ref{algo:committee}, which in turn is responsible 
for creating a set of $\Omega(\sqrt{n})$ search landmarks.
These search landmarks have high probability to be reached by any of the 
random walks originating from one of the storage landmarks that were 
previously created by the storage committee members.
In more detail, we can show that with high probability,  
$\Omega(\sqrt{n})$ search landmark nodes are from the same core set from 
which the $\Omega(\sqrt{n})$ storage landmarks have been chosen and 
therefore, within $O(\log n)$ rounds, a search landmark is very likely to 
get to know the id of one of the storage landmarks.

\begin{algorithm}[t]
  \begin{algorithmic}[1]
  \item[] Node $u$ issues a retrieval request in round $r_1$ for data $\I$.
  \STATE Node $u$ initiates Algorithm~\ref{algo:committee} to create a
  committee $\Com$ which will automatically dissolve itself after
  $\Theta(\log n)$ rounds. 
  \STATE Node $u$ instructs the committee members to execute
  Algorithm~\ref{algo:tree} and repeatedly create a landmark set of
  $\Omega(\sqrt{n})$ nodes.
  Every landmark node $w$ contacts all nodes of received samples and inquires
  about $\I$.
  If $\I$ is found, $w$ directly reports this to $u$.
  \end{algorithmic}
  \caption{Retrieval of a Data Item}
  \label{algo:search}
\end{algorithm}
\begin{theorem}[Data Retrieval] \label{thm:retrieval} \label{thm:search}
  Consider any round $r_1 \ge 2\tau$. There is a set $A$ of at least $n-o(n)$
  nodes, such that any available data item $\I$ can be retrieved by any $u 
  \in A$ via Algorithm~\ref{algo:search} in $O(\log n)$ rounds, with high 
  probability, in a network with churn rate up to $O(n/\log^{1+\delta}n)$ 
  per round.
\end{theorem}

\onlyLong{
\begin{proof}
  By assumption, item $\I$ is available in $r_1$, wich means that every 
  node $v \in \Core^{[r_1,r_1+\tau]}$ has probability at least 
  $\frac{1}{\Theta(\sqrt{n})}$ to be in $M^{r_1}_\I$.
  Moreover, by Corollary~\ref{cor:available}, item $\I$  will still be 
  available for a polynomial number of rounds with high probability.
  By Lemma~\ref{lem:tree}, we know that, after $O(\log n)$ rounds, the 
  committee created by $u$ has constructed a set $T$ of $\Omega(\sqrt{n})$ 
  nodes in $\Core^{[r_1,r_1+\tau]}$ that will report any encounter with a 
  landmark of $\I$ to $u$.

  For any $v \in \Core^{[r_1,r_1+\tau]}$, we know by 
  Lemmas~\ref{lem:bad_rand_walks} and \ref{lem:sufficient} that $v$ 
  receives at least one walk that originated from some
  $w \in \Core^{[r_1,r_1+\tau]}$ w.h.p., thus the probability of $v$ not 
  getting to know the id
  of a landmark node $M_i^{r_1}$ in round $r$ is at most
  $1-\frac{1}{\Theta(\sqrt{n})}$.
  Applying the same argument to each of the $\Omega(\sqrt{n})$ nodes in $T$, 
  shows that the probability of none of them finding a landmark node for $\I$ 
  is at most $\left(1-\frac{1}{\Theta(\sqrt{n})}\right)^{\Theta(\sqrt{n})}
  \le e^{-\Omega(1)}$.
  Note that, for the next $\tau (\in \Theta(\log n))$ rounds, we have the same 
  probabilities for the nodes in $T$ to encounter a landmark node of $\I$.
  (This nodes are in $\Core^{[r_1,r_1+\tau]}$ by assumption, thus they 
  will not be subjected to churn before round $r_1+\tau$.)
  It follows that, within $O(\log n)$ rounds, one node in $T$ will receive 
  a sample from a landmark of $\I$ and thus $u$ will be able to retrieve 
  $\I$ with high probability.
\end{proof}
}

\iffalse
\noindent{\bf Load Balancing:}
%
%
%
%
%
%
%
It is easy to see that our algorithm achieves almost load balancing through most individual nodes of the network.
In particular, a node $u$ (in the current {\sc Core}) that wants to store some item ${\cal I}$ receives
$\Theta(\log n)$ samples that are almost uniformly distributed among the nodes in the current {\sc Core}); 
a subset of these nodes is chosen to form a committee of ${\cal I}$, i.e., the nodes in which ${\cal I}$
will be stored. 
This implies that most nodes ($n-o(n)$) of the network receives comparable number of data items, and hence our 
algorithm achieves almost load balancing provided that stored data items are of comparable sizes.
In addition to load balancing, our algorithm also allows fast retrieval and persistence of ${\cal I}$ with high
probability. 

\fi

%
\onlyLong{
\subsection{Reducing the number of bits stored using erasure codes}\label{sec:erasure-codes}

We can further reduce the total  number of bits in storing large data items using the standard
technique of erasure codes. We next describe  how to incorporate such a technique in our scheme.

Given a data item ${\cal I}$ to be stored in the network, the algorithm described in the 
previous sections simply replicates ${\cal I}$ at a set of appropriate number of nodes
in the network. 
The drawback of replication is the consumption of a high amount of network bandwidth and storage capacity.
The other method is to apply erasure codes (e.g., Information Dispersal Algorithm (IDA) \cite{r89}) 
to encode a data item into a longer message such that a fraction of the data suffices to 
reconstruct the original data item.
In particular, when applying IDA to storage systems, a data item ${\cal I}$ of length $|{\cal I}|$ is 
divided into $L$ parts, each of length $|{\cal I}|/K$ so that every $K$ pieces 
suffice for constructing ${\cal I}$.
The total size of all piecies equals $L|{\cal I}|/K$, and hence IDA is space efficient 
since we can choose the blowup ratio $L/K$, that determines the space 
overhead incurred by the encoding process, to be close to $1$.

Here, we show that our algorithm described in Section \ref{sec:storage} 
can be simply modified to apply erasure codes for storing data items in the
network; however, the most challenge will be maintaining at least $K$ pieces of ${\cal I}$ under
node churn (the number of nodes storing pieces of ${\cal I}$ might be decreased rapidly with time).

Suppose that a node $u$ wants to carry out an insersion process of a data
item $\cal{I}$ in round $r_1$.
First, $u$ creates a committee $\Com^{r_1}$ of $h\log n$ members for $\cal{I}$ by executing 
Algorithm~\ref{algo:committee}, applies IDA to split $\cal{I}$ into $h\log n$ pieces, 
each of size $|{\cal I}|/((h-2)\log n)$, and then requests each member of
$\Com^{r_1}$ to store one of these pieces. 
Next, $u$ instructs the committee members to execute Algorithm~\ref{eq:treedepth} and repeatedly 
create landmark sets of $\Omega(\sqrt{n})$ nodes that will respond to retrieval requests of $\cal{I}$.
Now, consider round $r_2=r_1+\tau$, in which members of $\Com^{r_1}\cap V^{r_2}$ 
execute Algorithm~\ref{algo:committee} to construct a new committee of ${\cal I}$. 
We slightly modify the committee maintainace stage in Algorithm~\ref{algo:committee} as follows.
We first bound the size of $\Com^{r_1} \cap V^{r_2}$.
Note that the probability of any node $s \in \Core^{r_1}$ to be in
$\Com^{r_1}$ is bounded by $[\frac{1}{17n},\frac{3}{2n}]$.
Since the adversary is oblivous and the churn rate is ${O}(n/\log^{1+\delta} n)$, 
the probability of a node $v\in \Com^{r_1}$ is churned out
in a later round $r>r_1$ is at most 
\begin{equation}\label{bound-1}
\frac{3}{2n}\cdot \frac{4n}{\log^{1+\delta} n}=\frac{6}{\log^{1+\delta} n}.
\end{equation}
By taking a union bound on (\ref{bound-1}) over rounds in $[r_1,r_2]$, we get that 
\[
  \Prob{(v \in \Comm^{r_1}) \wedge (v \notin V^{r_2})} \le \frac{6}{\log^{\delta}n}.
\]
Let $X$ be the random variable determined by the number of nodes in
$\Com^{r_1}$ that are subjected to churn in $[r_1,r_2]$.
Since $|\Com^{r_1}| = h\log n$, it follows that
\[
  \Exp{X} \le 6h \log^{1-\delta} n.
\] 
By using a standard Chernoff bound, we get 
$\Prob{X \ge 2\log n > 6\Exp{X}} \le 2^{-2\log n},$
which shows that, with probability at least $1-\Theta(n^{-2})$, the size of
$\Com^{r_1}$ is reduced by at most $2\log n$ in $[r_1,r_2]$, i.e.,
\[
 |\Comm^{r_1}\cap V^{r_2}| \ge (h-2)\log n.
\] 
Let $c^{r_2}$ be the node defined in Algorithm~\ref{algo:committee}
which, by the definition, knows ids of all other nodes in $\Com^{r_1} \cap V^{r_2}$.
Therefore, with high probability, ${\cal I}$ can be reconstructed at $c^{r_2}$ in round $r_2+1$.
At round $r_2+2$, $c^{r_2}$ chooses $h\log n$ random walks that stopped at $c^{r_2}$ in round $r_2$ and
invites their source nodes to form the new committee in $r_2+3$.
At the same time (round $r_2+2$), $c^{r_2}$ reconstructs the original data item ${\cal I}$, 
replicates ${\cal I}$ by applying IDA, and then requests, along the invitation of joining the new committee, 
each candidates of the new committee to store one piece of the resulting parts.

To retrieve a data item ${\cal{I}}$, node $u$ interested in $\cal{I}$ 
creates a committee, and then requests the committee members to execute Algorithm~\ref{eq:treedepth} 
and repeatedly create a landmark set of $\Omega(\sqrt{n})$ nodes.
Every landmark node $w$ contacts all nodes of received samples and inquires
about $\cal{I}$.
If a piece of $\cal{I}$ is found, say at node $v$, then $w$ directly reports this to $u$.
Note that $v$ is a member of the committee storing $\cal{I}$, and hence knows the ids of all
other members of this committee.
This enables $u$ to contact the committee of $\cal{I}$ and to reconstruct the original item at $u$.
}

\section{Conclusion}
We have presented efficient algorithms for robust storage and retrieval of data items in a highly dynamic setting where a large number of nodes can be subject to churn in every round and the topology of the network is under control of the adversary.  An important open problem  is finding  lower bounds for the maximum amount of churn that is tolerable by any algorithm with a sublinear message complexity.  For random walks based approaches, we conjecture that there is a fundamental limit at $o(n/\log n)$ churn, for the simple reason that if churn can be in order $\Omega(n/\log n)$, the adversary can subject a constant fraction of the nodes to churn by the time a random walk has completed its course.  In this context, it will be interesting to determine exact tradeoff between message complexity and tolerable amount of churn per round.

\bibliographystyle{plain}
\bibliography{papers,papers1}

\begin{comment}
\onlyLong{
\normalsize
\newpage
\section*{Appendix}
\appendix
\input{appendix}
}
\end{comment}

\end{document}